\documentclass[conference]{IEEEtran}

\usepackage{amsmath}
\usepackage{amssymb}
\usepackage{cite}	
\usepackage{epsfig}
\usepackage{graphicx}
\usepackage{enumerate}
\usepackage{algorithm}
\usepackage{algorithmic}
\usepackage{hyperref}
\usepackage{color}
\usepackage{caption}
\usepackage{subcaption}

\newtheorem{theorem}{{\bf Theorem}}
\newtheorem{corollary}{{\bf Corollary}}
\newtheorem{lemma}{{ \bf Lemma}}
\newtheorem{definition}{{\bf Definition}}

\newtheorem{proposition}{Proposition}

\def\Reals{\Bbb{R}}

\newtheorem{assumption}{Assumption}

\newcommand{\beq}{\begin{equation}}
\newcommand{\eeq}{\end{equation}}

\newcommand{\bdisp}{\begin{displaymath}}
\newcommand{\edisp}{\end{displaymath}}

\newcommand{\beqarr}{\begin{eqnarray}}
\newcommand{\eeqarr}{\end{eqnarray}}

\newcommand{\bmlt}{\begin{multline}}
\newcommand{\emlt}{\end{multline}}

\newcommand{\beqarrn}{\begin{eqnarray*}}
\newcommand{\eeqarrn}{\end{eqnarray*}}

\newcommand{\benum}{\begin{enumerate}}
\newcommand{\eenum}{\end{enumerate}}

\newcommand{\bit}{\begin{itemize}}
\newcommand{\eit}{\end{itemize}}

\newcommand{\bc}{\begin{center}}
\newcommand{\ec}{\end{center}}

\newcommand{\bdes}{\begin{description}}
\newcommand{\edes}{\end{description}}

\newcommand{\bfig}{\begin{figure}}
\newcommand{\efig}{\end{figure}}

\newcommand{\bemq}{\begin{quote} \begin{em}}
\newcommand{\eemq}{\end{em} \end{quote}}

\newcommand{\bmp}{\begin{minipage}}
\newcommand{\emp}{\end{minipage}}

\newcommand{\eqn}[1]{(\ref{#1})}

\newcommand{\floor}[1]{\left\lfloor{#1}\right\rfloor}
\newcommand{\ceil}[1]{\left\lceil {#1} \right\rceil}

\newcommand{\indic}[1]{{1}_{\left\{{#1}\right\}}}

\newcommand{\kth}{^{{\mathrm{th}}}}

\newcommand{\ie}{{i.e.}}

\newcommand{\eg}{{ e.g.}}

\newcommand{\iid}{{i.i.d.}}

\newcommand{\bsp}{\begin{slide*}}
\newcommand{\esp}{\end{slide*}}
\newcommand{\bsl}{\begin{slide}}
\newcommand{\esl}{\end{slide}}

\newcommand{\Scale}[2][4]{\scalebox{#1}{$#2$}}%

\newcommand{\e}{{e}}
\begin{document}

\author{\IEEEauthorblockN{Virag Shah}
\IEEEauthorblockA{Microsoft Research-Inria Joint Centre\\
 Palaiseau, France, 91120\\
virag.shah@inria.fr}
\and
\IEEEauthorblockN{Anne Bouillard}
\IEEEauthorblockA{ENS/INRIA\\
Paris Cedex 05, France, 75230 \\
anne.bouillard@ens.fr}
\and
\IEEEauthorblockN{Fran\c{c}ois Baccelli}
\IEEEauthorblockA{The University of Texas at Austin\\
 Austin, TX, USA, 78712 \\
francois.baccelli@austin.utexas.edu}
}


%

\title{
Delay Comparison of Delivery and Coding Policies
in Data Clusters}

\maketitle

\begin{abstract}
A key function of cloud infrastructure is to store and deliver diverse
files, \eg, scientific datasets, social network information, videos, etc. 
In such systems, for the purpose of fast and reliable delivery,
files are divided into chunks, replicated or erasure-coded,
and disseminated across servers. It is neither known in general
how delays scale with the size of a request 
nor how delays compare under different 
policies for coding, data dissemination, and delivery. 

Motivated by these questions, we develop and explore a set of evolution 
equations as a unified model which captures the above features.
These equations allow for both efficient simulation and mathematical analysis 
of several delivery policies under general statistical assumptions. 
In particular, we quantify in
what sense a workload aware delivery policy performs better
than a workload agnostic policy. Under a dynamic or stochastic setting, 
the sample path comparison of these policies does not hold
in general. The comparison is shown to hold under the weaker
increasing convex stochastic ordering,
still stronger than the comparison of averages.

This result further allows us to obtain insightful
computable performance bounds. For example,
we show that in a system where files are divided into chunks of equal size,
replicated or erasure-coded, and disseminated across servers at random,
the job delays increase sub-logarithmically in the request size
for small and medium-sized files but linearly for large files. 


\end{abstract}



\section{Introduction}\label{sec:Intoduction}

Modern cloud computing infrastructures feature several clusters each of which consists of thousands of highly interconnected servers which collectively run and serve diverse computing applications \cite{VPK15,SKR10,GHJ09}. An important aspect of these clusters is to collectively store and deliver Internet scale data/files. Key design challenges for such systems include placement of files across servers and an algorithm for the swift delivery of dynamically arriving file requests.  A common practice towards file placement is to divide each file into chunks of fixed size, which could then potentially be replicated/coded, and to disseminate them across the servers \cite{GGL03,FLP10}. This can potentially reduce delays in delivering large files since the delivery algorithm could now aggregate the service rate from multiple servers. 

To gain intuition, consider some hypothetical scenario where, for the placement of each file, one is allowed to use variable and arbitrarily small (possibly fractional) chunk sizes. Then, for a system of $m$ servers, one could divide each file of size $\nu$ bits into $m$ different chunks, each of size $\nu/m$ bits. Suppose that the service/delivery rate at each sever is $\mu$ bits/sec and that there is no other network bottleneck. Then, the minimum achievable delay in serving a download request for a file of size $\nu$ is $\frac{\nu}{m \mu}$, which is possible only if no other request is present in the system. 

However, delays which scale inverse linearly with $m$ clearly cannot be achieved for each file if there is a limit to the minimum chunk size. For example,  suppose that each file of size $\nu$ is divided into $\ceil{\frac{\nu}{c}}$ chunks of size $c$.  Then,  if there is only one download request in the system at a given time, for any file of size less than $cm$, a download delay equal to $\frac{c}{\mu}$ can be achieved. Whereas for files several times larger than $cm$ bits, the delay is still of the order of $\frac{\nu}{m \mu}$ under isolation. However, for a system with diverse files and fixed chunk size, it is not directly clear what the delays are under stochastic loads. In such a setting, how do delays relate with the size of a requested file? Do replication of chunks or erasure coding help in reducing delays? What is the impact of dynamic load-balancing? These are some of the questions we address in this work. 

\underline{\fontsize{11pt}{11pt}\selectfont Contributions:} We provide a stochastic model which encapsulates the key features of the content delivery process in a highly interconnected cluster of servers and allows us to compare several different policies as well as to obtain explicit performance bounds. In particular, our model captures the following aspects: 

{\em Dissemination policy:} We allow each file to be divided into chunks of a given size. Thus, a larger file is divided into larger number of chunks.  These chunks are most often encoded to obtain code blocks of the given size, as explained below. The code blocks are then disseminated across servers in a randomized fashion to ensure that the load across servers is balanced. 

{\em Coding policy:} Suppose that a file is divided into $k$ chunks. For each $k \ge 1$, these chunks are coded into $\alpha_k$ code blocks for some $\alpha_k\ge k$ via MDS (maximum distance separable) erasure codes \cite{DGW10,PlH13}. These codes are designed such that the original $k$ chunks can be exactly recovered from any $k$ out of the $\alpha_k$ code blocks. This allows additional flexibility towards dynamically balancing load across servers as the requests arrive over time, as explained below. 

{\em Delivery policy:} Upon the arrival of a request for a file with $k$ chunks, a request is sent to a subset of servers to obtain $k$ out $\alpha_k$ associated code blocks. The servers serve the block requests in FCFS fashion. We allow dynamic load balancing policies such as {\em Water-filling} and {\em Batch Sampling} policies (defined below) which favor a subset of the set of servers with lower instantaneous loads to balance the server workload as well as to achieve lower request delays.  

We propose a comprehensive model for
this class of systems, with the potential of representing
all such policies under certain diversity and symmetry assumptions
on the file sizes and the loading policy. This model consists of a set of
evolution equations which allow for both efficient
simulation and mathematical analysis under general statistical
assumptions. In particular, we are able to show the following: 

\begin{enumerate}
\item We compare the evolution of workloads under three different delivery policies: namely, water-filling ($\mathcal{WF}$), batch sampling ($\mathcal{BS}$),
and a randomized policy called {\em Balanced Random} ($\mathcal{BR}$).
We show that, for a given workload at each server, $\mathcal{WF}$ is optimal in the sense that upon a new arrival, it achieves `the most balanced' workload as compared to any other policy. Further, $\mathcal{BS}$ is somewhere in between
$\mathcal{WF}$ and $\mathcal{BR}$ in this respect. 
\item We show that $\mathcal{WF}$ and $\mathcal{BS}$ achieve more favorable workload distributions and lower delay distributions as compared to $\mathcal{BR}$ in the sense of `increasing convex order', which in turn implies that the former policies achieve better performance not only in expectation but in higher moments as well. 
\item We provide an upper bound for the delay in delivering a file as a function of its size, under a scenario where the requests form a mix of diverse file sizes. Our bound reveals the relative impact of the local dynamics at an individual server and that of the global view of server workloads seen by an arrival. We also provide new scaling laws on the behavior of delays under such a scenario. 
\item Using simulations we analyze the impact of the key options and 
parameters, including the delivery policy, the coding options and
the chunck size. 
We identify two fundamental regimes, the logaritmic regime when the file sizes are such that no two chunks are stored on the same server,  and the linear regime when files have a number of chunks that exceeds the number of servers.
We show that in the logarithmic regime, the gains of dynamic load balancing via $\mathcal{WF}$ and $\mathcal{BS}$ are significant even when the coding rate is small. We also show that our product form bound on the delivery latency is tight when requests have a moderate size. 
\end{enumerate}

\underline{\fontsize{11pt}{11pt}\selectfont Related Work:}
Recently there has been significant interest towards developing scalable
performance models and analysis for content delivery systems with low delays.
For example, the work in \cite{ShV15_TON,ShV16_Ind,ShV16_Het} 
exploits server parallelism via ``resource pooling'', that is
multiple servers are allowed to work together as a pooled
resource to meet individual download requests.
The pools of servers associated with different requests may overlap,
so the sharing of server resources across classes is
done via a fairness criterion.
Under a scenario where the size of resource pools is {\em limited} (\ie
$o(m)$), it is shown that the gains of resource pooling and
load-balancing can be achieved simultaneously.


An alternate approach considered in the literature is to split
a download request into multiple parts, for example, into requests
for individual chunks, and achieve server parallelism by employing
different servers for different parts
\cite{Vve98,LRS16}.
Further, sophisticated coding policies are employed to
achieve flexibility in server choices \cite{LRS16}.
Under the assumption that the number of servers available for each request
is {\em limited}, the works in \cite{Vve98,LRS16} are able to use mean-field
based arguments to study performance as the number of servers $m$ tends to
infinity. Several other works also study queuing models under coding based
techniques via heuristics or bounds, \eg \cite{SLR14,JLS14,JDP05,LiK14},  but these
are not scalable for our purposes. 


We depart from the above approaches in that we are interested in developing
performance models for a regime where we obtain a maximum gain from server
parallelism without restricting ourselves to limited resource pools or
a limited number of available servers for each request.
Given a lower bound on the chunk size, we divide each file into
a maximum number of chunks and disseminate them across several
servers, potentially $\Omega(m)$ servers for large files. Thus, we allow
$\Omega(m)$ servers to take care of a request in parallel.
We allow diverse file sizes and
provide a delay bound which is a function of the file size.

In terms of tools used, we model the system dynamics via an evolution equation which is a generalization of the Kiefer and Wolfowitz recursion for workloads in $G/G/s$ queues \cite{BaB03}, which allows us to go beyond exponentiality assumptions for file-size requests. We use coupling arguments to compare different policies. Coupling has been used to compare several queueing systems in past, \eg, see for example~\cite{ShV16_Het,BMS89}. Further, to provide explicit bounds on delays, we use the notion of association of random variables, which is a property that has had several applications in queueing systems and beyond \cite{BaB03,MuS02}.

\underline{\fontsize{11pt}{11pt}\selectfont Organization:} In Section \ref{sec:SystemModel} we provide our system model and develop the evolution equations. In Section \ref{sec:Comparisons} we provide results comparing various dynamic load balancing policies via coupling arguments. In Section \ref{sec:Bounds} we give performance bounds based on the notion of association of random variables. In Section \ref{sec:RandomChunkSize} we consider a scenario where the chunk size may be different for different files.  In Section \ref{sec:sim-rec} we provide simulation results and numerical evaluations.  We conclude in Section~\ref{sec:Conclusions}. Some proofs of technical nature are provided in the Appendix.


\section{System Model}\label{sec:SystemModel}

We consider a system with $m$ servers, indexed $1,2,\ldots,m$. The system consists of a very large number (several orders of magnitude larger than $m$) of diverse files.  We assume that the size of each file is an integer multiple of $c$ bits. Each file is divided into chunks of size $c$ bits each. These chunks are encoded before being placed across servers, as explained below.

For each positive integer $k$ we use an MDS erasure code of rate $k/\alpha_k$, where $\alpha_k$ is an integer greater than or equal to $k$. Such a code is called $(\alpha_k,k)$ MDS code in coding theory~\cite{DGW10}. Thus, equivalently, each file of size $kc$ is divided into $k$ chunks and encoded into $\alpha_k$ code blocks of size $c$ bits each. The MDS erasure codes may serve various practical purposes. Only the following property is relevant for our purposes: for a file of size $kc$, it is possible to recover the entire file by downloading any $k$ out of the $\alpha_k$ code blocks.

For each file of size $kc$, the associated $\alpha_k$ code blocks are placed across servers as follows. If $\alpha_k < m$, then we choose $\alpha_k$ among the $m$ servers uniformly at random and place a distinct code block across each of these servers. Else, we place $\floor{\frac{\alpha_k}{m}}$ distinct blocks on each server and for the remaining $\alpha_k - m\floor{\frac{\alpha_k}{m}}$ blocks we choose that many servers uniformly at random.

We assume that the blocks are placed across servers as described above at time $t=-1$.
The placement of blocks is kept fixed since then. From time $t=0$, file download
requests arrive as per an independent Poisson point process $ \Pi$ with 
rate $\lambda$.  Let $\{t_0, t_1, \ldots\}$ be the points of $ \Pi$.

Consider a probability mass function $\pi = (\pi_k: k\in \mathbb Z_+)$.
Each request arrival corresponds to a file of size $ck$ bits with probability $\pi_k$
independently of all other arrivals.
Let $\kappa_n$ be the number of chunks for the file requested at time $t_n$.
Thus, $\{\kappa_n\}_0^\infty$ is a sequence of discrete i.i.d. random variable with p.m.f.\ $\pi$. 

For each $n$, let $a_n \in \mathbb Z_+^m$ represent the placement of the file requested
at time $t_n$, in the following sense: for each server $i$ the entry $a_n^i$ represents the
number of coded blocks placed on server $i$ that correspond to the file requested upon the $n\kth$ arrival.
Thus, for each $n,k \in \mathbb Z_+$, if $\kappa_n = k$, then $|a_n| = \alpha_k$. 
We call $\{a_n\}_0^\infty$ the sequence of placement vectors. 

Let $\nu = c \sum_{k=0}^{\infty}{k\pi_k}$ denote the mean file-size in bits.
Let $\rho = \lambda \nu/m$ denote the per server load in bits/sec. 


\begin{assumption}[Symmetry in load across servers]
Due to the randomized placement of blocks, for a very large number of files,
the load across servers is approximately symmetric. Thus we model the symmetry in load via
symmetry in request arrivals as follows: given $\kappa_n = k$, $a_n$ is chosen uniformly 
at random from each of its feasible realizations.
Equivalently, given $\kappa_n = k$, the entry $a_n^{i}$ is equal to 
$\floor{\frac{\alpha_k}{m}} + 1$ for $\alpha_k - m\floor{\frac{\alpha_k}{m}}$
servers chosen uniformly at random and it is equal to $  \floor{\frac{\alpha_k}{m}}$ for the rest of the servers.

Making such a symmetry assumption to obtain insightful results
is a common practice, see \eg\ \cite{Vve98,ShV15_TON,LRS16}.
While, in general, a system with a finite number of 
files may not be symmetric, we believe that this is a good
approximation especially when the number of files is an
order of magnitude larger than the number of servers.

\end{assumption}

We will not discuss server memory capacity issues here as
this is not needed.
Note however that such a randomized placement results
into concentration of memory usage at each server.


%
%

\begin{figure*}
{\small
        \centering 
                \begin{subfigure}[b]{0.65\columnwidth}
		\centering
 		\def \svgwidth{0.8\columnwidth}
 		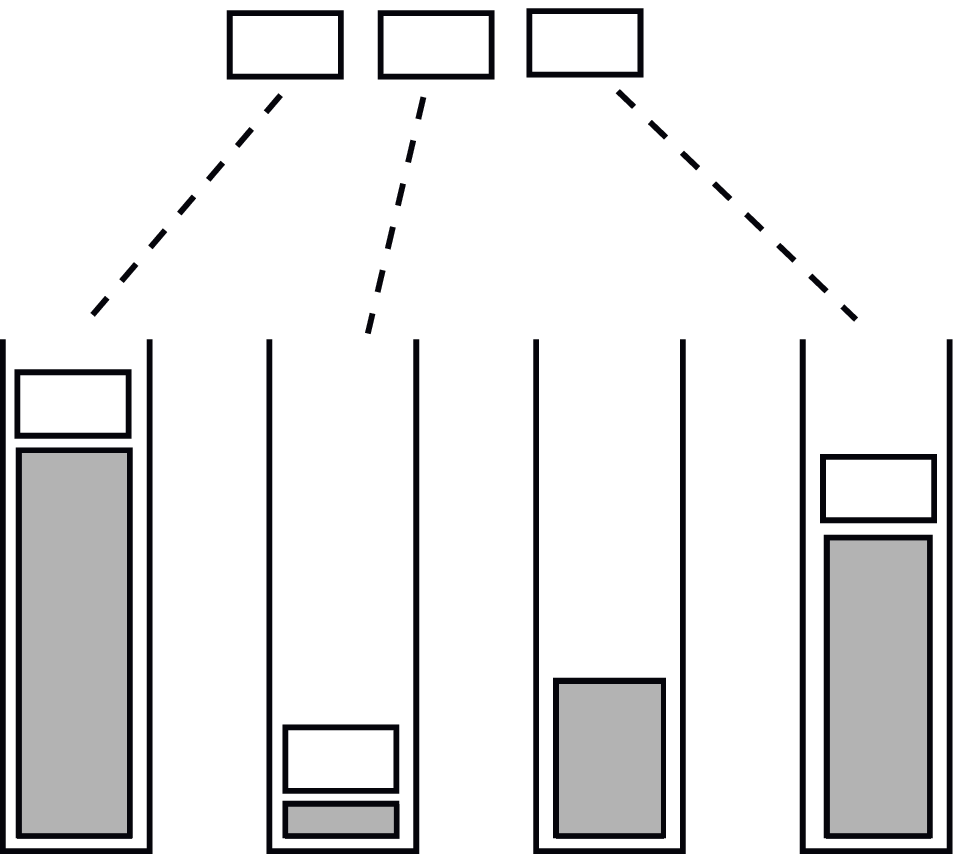          
	     \caption{Balanced Random}
        \end{subfigure}
        \begin{subfigure}[b]{0.65\columnwidth} 
		\centering
 		\def \svgwidth{0.8\columnwidth}
 		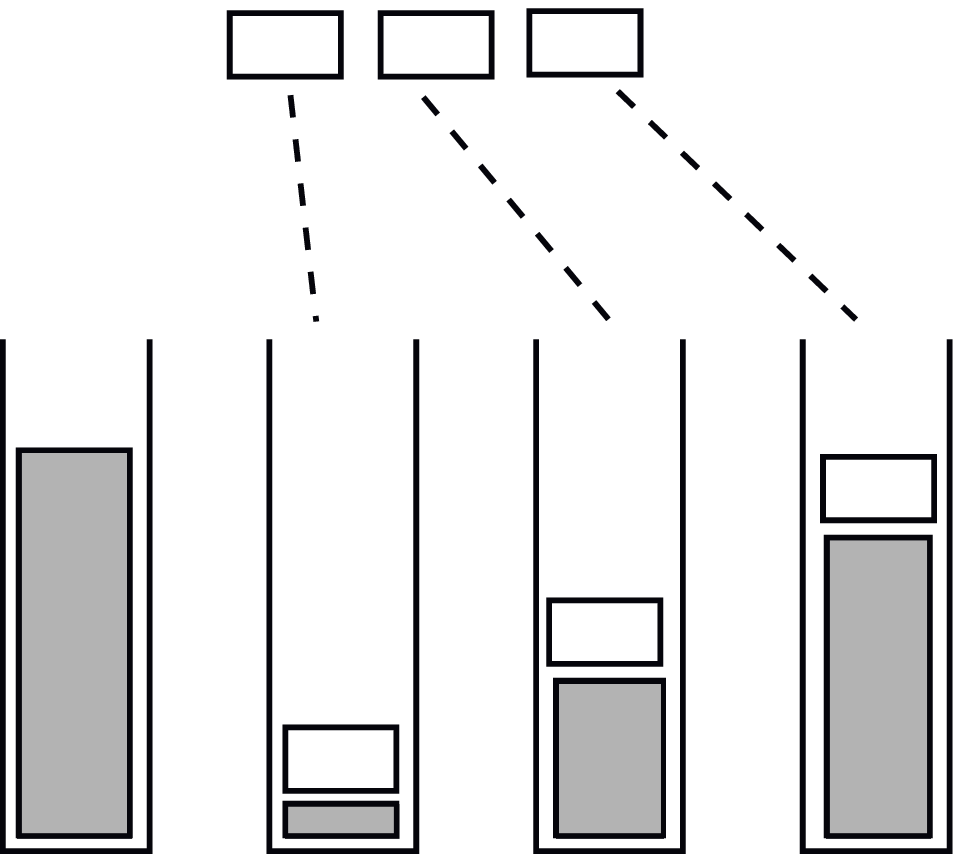          
	     \caption{Batch Sampling}  \label{fig:BS}
        \end{subfigure}
                \begin{subfigure}[b]{0.65\columnwidth}
		\centering
 		\def \svgwidth{0.8\columnwidth}
 		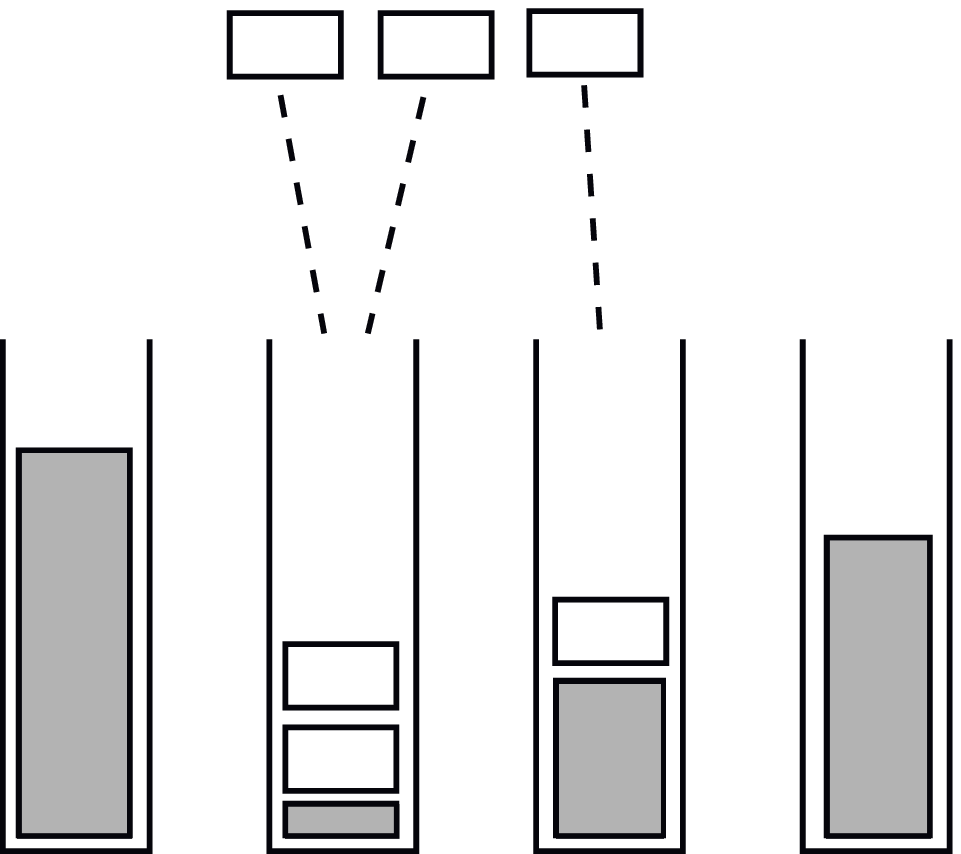           
	     \caption{Water-filling} \label{fig:WF}
        \end{subfigure}
   \caption{Illustration of different dynamic delivery policies upon the $n\kth$ arrival; $m=4$, $k=3$, $\alpha_k=5$, $a_n = (1,2,1,1)$. }  \label{fig:DynamicLoadBalancing}
 }
\hspace{-.5cm}
\end{figure*}

{\em  Delivery policy:} Upon each arrival, we load servers with requests for coded blocks via a delivery/routing policy as described below. Each server serves its block requests in FCFS fashion at rate $\mu$, \ie, it delivers a code block at the rate of $\mu$ bits per second.  Recall, due to our use of MDS codes, if $\kappa_n = k$ then the system only needs to deliver $k$ out of the $\alpha_k$ associated blocks for the $n\kth$ arrival. We let $s_n$ denote the $\mathbb Z_+^m$ valued random variable where $s_n^i$ is the number of blocks requested from server $i$ upon the $n\kth$ arrival. Thus, we have $|s_n| = \kappa_n$ and $s_n \le a_n$ for each $n$. We call $\{s_n\}$ the sequence of {\em routing vectors}. Following are some of the admissible routing policies, each resulting into possibly different sequences of routing vectors.

{\em Balanced Random Policy ($\mathcal{BR}$):} For each $n,k$, if $\kappa_n = k$, then request $ \floor{\frac{k}{m}}$ blocks from each server and, for the remaining $k - m\floor{\frac{k}{m}}$ blocks, choose the same number of servers at random from the remaining $\min\left(\alpha_k - m  \floor{\frac{k}{m}}, m\right)$ servers having an additional block. More formally, suppose $\kappa_n = k$. Let
 $k' = k - m\floor{\frac{k}{m}}$ and $a'_n = a_n - \floor{\frac{k}{m}}{\bf 1}$. From the set $\{i: {a'}^i_n > 0\}$ choose a subset of size $k'$ at random. Let $s^i_n$ be equal to $ \floor{\frac{k}{m}} + 1$ for each $i$ in this subset and $\floor{\frac{k}{m}}$ for others.

The following two policies take a routing decision upon the $n\kth$ arrival based on the instantaneous workloads at different servers at time $t_n^{-}$. 

{\em Batch Sampling Policy ($\mathcal{BS}$):} This is a workload dependent
policy. The workload at a server at any given time is the number of
bits requested from the server and which are not yet served. 
Of the required $k$ blocks, request $ \floor{\frac{k}{m}}$ blocks
from each server and for the remaining $k' = k - m\floor{\frac{k}{m}}$ blocks, choose the $k'$ servers with least instantaneous workload from the remaining $\min\left(\alpha_k - m  \floor{\frac{k}{m}}, m\right)$ servers having an additional block. More formally, suppose that the workload at the servers at time $t_n^-$ is $w = (w^i: i=1,\ldots,m)$ and that $\kappa_n=k$. Let $k' = k - m\floor{\frac{k}{m}}$ and $a'_n = a_n - \floor{\frac{k}{m}}{\bf 1}$. 
Let $i_1,i_2,\ldots,i_{k'}$  be given recursively as follows: let $i_1 =  \mathop{\arg\min}_{   i:  {a'}^i_n > 0 } w^i$, and for $l =2,\ldots,k'$ let $i_l = \mathop{\arg\min}_{   i:  {a'}^i_n > 0,  i \neq i_1,\ldots,i_{l-1}  } w^i. $ 
Then, we have  $s_n^i = \floor{\frac{k}{m}}+1$ for each $i \in \{i_1,i_2,\ldots,i_{k'}\}$ and $s_n^i = \floor{\frac{k}{m}}$ for 
$i \notin \{i_1,i_2,\ldots,i_{k'}\}$.

{\em Water-filling Policy ($\mathcal{WF}$):} This is also a workload
dependent policy. If $\kappa_n=k$, then at time $t_n$, we take a
routing decision for $k$ block requests defined sequentially as follows. 
Among the servers which store at least one of the $\alpha_k$ blocks for the associated file, choose the server with minimum workload. If there are multiple such servers, choose one at random. Request a block from this server and update its workload, i.e., add $c$ to its existing value. We now have to choose $k-1$ blocks among the $\alpha_k-1$ remaining code blocks, for which we repeat the above procedure, see Fig.~\ref{fig:DynamicLoadBalancing}. 

More formally, suppose that the workload at the servers at time $t_n^-$ is $w$ and that $\kappa_n=k$. 
 Then, let $j_1,j_2,\ldots,j_k$ be recursively given as follows: $j_1 =  \mathop{\arg\min}_{   i:  {a}^i_n > 0 } w^i$, and for $l =2,\ldots,k$ let
 $$ j_l = \mathop{\arg\min}_{i: a_n^i - \sum_{l' = 1}^{l-1} \indic{i=j_{l'}} >0}  w^i + c \sum_{l' = 1}^{l-1} \indic{i=j_{l'}}.$$
For $i=1,\ldots,m$, let $e_i$ represent the vector in $\mathbb R^m$ with $i\kth$ entry equal to $1$ and other entries equal to $0$. Then, under the $\mathcal{WS}$ policy we have $ s_n = \sum_{l=1}^m \e_{j_l}$.

One would guess that $\mathcal{WF}$ is the most egalitarian policy, \ie, it attempts at spreading the arriving load to servers with lower instantaneous workloads, and $\mathcal{BS}$ is somewhere in between $\mathcal{WF}$ and $\mathcal{BR}$ in egalitarianism.  We will corroborate these intuitions in the next section.

Note that we do not allow policies which depend explicitly on the server indices. More concretely, if server indices are permuted at time $t=0^-$, the choice of servers upon each arrival is permuted in the corresponding fashion.

 Recall that the routing vector $s_n$ for each $n$ is such that $|s_n|$ is chosen independently with distribution $(\pi_k: k\in \mathbb N)$, while its entries depend on the workload at the servers at time $t_n^-$ and on the delivery policy. 
Due to symmetry in file placement (modeled via symmetry in request arrivals) and the above mentioned restriction on the delivery policies, we have that $\{s_n\}_0^{\infty}$ are exchangeable random vectors in the following sense: upon permutation of server indices the distribution of the sequence $\{s_n\}$ remains unchanged.

 Let $\{\tau_n\}_0^\infty$ be inter-arrival times, \ie, $\tau_n = t_{n+1} - t_n$ for each $n$.   
Let $\{W_n\}_0^\infty$ be a sequence of $\mathbb{R}_+^m$ valued random variables representing the workload seen by $n\kth$ arrival, \ie, the workload at different servers at time $t = t_n^-$.  Then we have $W_0 = {\bf 0}$ and 
\begin{equation}\label{eq:Main_Recursion}
W_{n+1} = (W_n + c s_n - \mu \tau_{n}  {\bf 1})^+  ,  \;\;\; n=0,1,\ldots
\end{equation}
where ${\bf 1}= (1,1,\ldots,1)$, and $$(x^1,\ldots, x^m)^+ = (\max(x^1,0), \ldots, \max(x^m,0)).$$

The {\em delay} of the $n\kth$ request is then:
$$D_n = \mathop{\max}_{i: s_n^i > 0}  W_n^i+ c s_n^i, \;\;\; n=0,1,\ldots$$

As mentioned earlier, we are mainly interested in the case where $c$
is a constant since we want to obtain a maximum gain from server
parallelism. However, one can envisage a scenario where different
requests/files use different chunk sizes.
This can be incorporated in our model as follows:
we have $W_0 = {\bf 0}$ and 
\begin{equation}\label{eq:Main_Recursion2}
W_{n+1} = (W_n + c_n s_n - \mu \tau_{n}  {\bf 1})^+  ,  \;\;\; n=0,1,\ldots,
\end{equation}
where the random variables $\{c_n\}_0^\infty$ are $\mathbb R_+$ valued and
\iid. Note that in this extension, for a given file, the chunks are
still of equal sizes. For most parts of the paper, we will use
recursion \eqn{eq:Main_Recursion}. We will nevertheless
discuss and analyze
\eqn{eq:Main_Recursion2} in Section~\ref{sec:RandomChunkSize}.

\section{Comparison of Delivery Policies}\label{sec:Comparisons}

In this section we compare the server workloads and the
request delays under different delivery policies.
We use coupling arguments to compare systems adopting
different delivery policies.
In particular, we couple the request arrival process
as well as the sequence of routing vectors in each system.
We then study and compare the evolution of server
workloads $\{W_n\}_0^\infty$ in the respective systems. 

For comparing the workloads of different systems,
we use stochastic submajorization and stochastic 
dominance in the increasing convex order sense,
which are briefly introduced in the first subsection.
While the former is more amenable to compare the loading
under different policies subject to a given initial condition,
the later allows us to propagate the comparison result and also
to compare delays (recall that the delay of a request is the $\max$
of the delays in downloading individual blocks).

\subsection{Order statistics and stochastic orders}
The notation and concepts listed below are
borrowed from \cite{MOA11} and \cite{MuS02}. 

For all vectors $z \in \mathbb R^{m}$, let $z^{(1)}, z^{(2)},\ldots,z^{(m)}$
represent its entries in increasing order. 

We say that a function $\phi:\mathbb R^m \to \mathbb R$ is symmetric if for all
$x \in \mathbb R^m$ and its permutation $x' \in \mathbb R^m$, we have $\phi(x) = \phi(x')$. 

For two vectors  $x,y \in \mathbb R^{m}$, we say that $x$ is majorized by $y$,
which is denoted by $x\prec y$, if $\sum_{i=1}^{m} x^i = \sum_{i=1}^{m} y^i $
and $\sum_{i=1}^{l} x^{(i)} \ge  \sum_{i=1}^{l} y^{(i)} $ for $l=1,2,\ldots,m-1$.
Intuitively, if $x \prec y$, then $x$ is `more balanced' than $y$.
For example, in $\mathbb R^m$, we have
$(1,1,\ldots,1) \prec (\frac{m}{2},\frac{m}{2},0,\ldots,0) \prec (m,0,0,\ldots,0)$. 

We say that $x$ is submajorized by $y$, which is denoted by $x\prec_s y$,
if $\sum_{i=l}^{m} x^{(i)} \le  \sum_{i=l}^{m} y^{(i)} $ for $l=0,1,2,\ldots,m-1$. 

We say that a function $\phi: \mathbb R^{m} \to \mathbb R$ is Schur-convex
if, for all $x$ and $y$ such that $x\prec y$, we have $\phi(x) \le \phi(y)$.
One can check that a function $\phi$ is Schur-convex and increasing if and only if (iff),
for all $x$ and $y$ such that $x\prec_s y$, we have $\phi(x) \le \phi(y)$. Further,
Schur-convex functions are symmetric since the property $x\prec_s y$
depends only on the ordered entries of $x$ and $y$.  

Consider two random vectors $X$ and $Y$. We say that $X$ is stochastically dominated by $Y$,
which is denoted by $X \le^{st} Y$, if, for all increasing functions $g$, we have $E[g(X)] \le E[g(Y)]$. 
A classical result (Strassen's theorem) states that $X \le^{st} Y$ iff there exist
random vectors $\tilde X$ and $\tilde Y$ such that $X$ and $\tilde X$ are identically distributed,
$Y$ and $\tilde Y$ are identically distributed, and $\tilde X \le \tilde Y$ w.p.\ $1$.

For two random vectors $X$ and $Y$ we say that $X$ is stochastically submajorized by $Y$, 
which is denoted by $X \prec^{st} Y$ if, for all Schur-convex functions $\phi$, we have
$E[\phi(X)] \le E[\phi(Y)]$. We have $X \prec^{st} Y$ iff there exist random vectors
$\tilde X$ and $\tilde Y$ such that $X$ and $\tilde X$ are identically distributed,
$Y$ and $\tilde Y$ are identically distributed, and $\tilde X \prec^{st} \tilde Y$ w.p.\ $1$.

Similarly, for the random vectors $X$ and $Y$ we say that $X$ is stochastically submajorized by $Y$, 
which is denoted by $X \prec_s^{st} Y$, if, for all increasing Schur-convex functions
$\phi$, we have $E[\phi(X)] \le E[\phi(Y)]$. 
Again, $X \prec_{s}^{st} Y$ iff there exist random vectors
$\tilde X$ and $\tilde Y$ such that $X$ and $\tilde X$ are identically distributed,
$Y$ and $\tilde Y$ are identically distributed, and $\tilde X \prec_s^{st} \tilde Y$ w.p.\ $1$.

For the random vectors $X$ and $Y$, we say that $X$ is stochastically dominated
by $Y$ in the increasing convex order sense, which is denoted by $X \le^{icx} Y$, if,
for all increasing convex functions $g$, we have $E[g(X)] \le E[g(Y)]$. 

For $i=1,\ldots,m$, let $e_i$ denote the vector in $\mathbb R^m$ with $i\kth$ entry
equal to $1$ and other entries equal to $0$. For any vector $x$ we let $|x|$ represent
the sum of the absolute values of its entries.

The following lemma is proved in the Appendix.

\begin{lemma}\label{lemma:MajImpliesIcx}
Consider $\mathbb R^m$ valued exchangeable random variable $X$ and $Y$. If $X \prec_s^{st} Y$ then we have $X \le^{icx} Y$.
\end{lemma}

\subsection{Comparison of Policies}

A delivery policy can be seen as a form of load balancing.
Intuitively, a more egalitarian load balancing should 
achieve more balanced overall workloads. For instance, recall
the policies $\mathcal{WF}$, $\mathcal{BS}$, and $\mathcal{BR}$
defined in the System Model.
The following theorem says that, given a workload vector $W$, $\mathcal{WF}$
is the most egalitarian policy while $\mathcal{BS}$ is somewhere in between
$\mathcal{WF}$ and $\mathcal{BR}$. For a proof, see the Appendix.

\begin{theorem}\label{thm:Egalitarianism}
 Suppose an arrival into the system sees the workload $W$, where $W$ is an $\mathbb R^m$
valued random variable. Let $s^{\mathcal{WF}}$, $s^{\mathcal{BS}}$, $s^{\mathcal{BR}}$, and $s'$
be the routing vectors associated with $\mathcal{WF}$, $\mathcal{BS}$, $\mathcal{BR}$, and
an arbitrary routing policy, respectively.  Then, the following holds.
$$W+ cs^{\mathcal{WF}} \prec^{st} W+cs' \;\;  \text{ and } 
\;\;  W+ cs^{\mathcal{BS}} \prec^{st} W + cs^{\mathcal{BR}}.$$
Further, if $W$ is an exchangeable random vector, then we have
$$W + cs^{\mathcal{WF}} \le^{icx} W + cs^{\mathcal{BS}} \le^{icx} W + cs^{\mathcal{BR}}.$$
\end{theorem}

Thus, for a given workload at $n$, a system under $\mathcal{WF}$ or $\mathcal{BS}$
achieves a more balanced workload in the $\prec_s$ sense at $n+1$ 
as compared to $\mathcal{BR}$. However, the resulting workloads might be different.
Starting with $W_0 = {\bf 0}$, to be able to claim that an ordering holds for each $n$,
one needs to argue that it propagates. For this we additionally need the 
monotonicity property of $\mathcal{BR}$ given in the lemma below. For a proof, see the Appendix. 

\begin{lemma}\label{lemma:MonotonicityOfBR}
Consider random vectors $W$ and $W'$ such that $W \le^{icx} W'$.
Let $s$ and $s'$ be the routing vectors as per the $\mathcal{BR}$ policy for
$W$ and $W'$ respectively. Then, $W + cs \le^{icx} W' + cs'.$
\end{lemma}

%
%
%

The following theorem establishes that the $\mathcal{WF}$ and $\mathcal{BS}$ policies achieve 
`more balanced and lower' workloads across servers as compared to $\mathcal{BR}$ in a strong sense.
For a proof, see the Appendix.

\begin{theorem}\label{thm:Workload_comparision}
Consider a system which starts empty.  The workload under policies $\mathcal{WF}$, $\mathcal{BS}$, and $\mathcal{BR}$ satisfy the following:
$$W_n^{ \Scale[0.6]{\mathcal{WF}} } \le^{icx}  W^{ \Scale[0.6]{\mathcal{BR}} }_n \text{ and } W_n^{ \Scale[0.6]{\mathcal{BS}} } \le^{icx}  W^{ \Scale[0.6]{\mathcal{BR}} }_n \text{ for } n=0,1,\ldots $$
\end{theorem}

\begin{proof}
%
%
We show the comparison result for a system with $\mathcal{BS}$ and a system with $\mathcal{BR}$; the argument for comparison for $\mathcal{WF}$ and $\mathcal{BR}$ is analogous.

Suppose that the two systems are fed with arrivals as given by the same point process $\Pi$. Thus, the sequence of interarrival times $\{\tau_n\}_0^\infty$ is the same for both systems. 

For ease of notation let $W_n, s_n$ represent the vectors associated with $\mathcal{BS}$ with their usual meaning, and let $W'_n, s'_n$ represent those associated with $\mathcal{BR}$. $W_0 \le^{icx}  W'_0$ holds trivially since both systems start empty. 
Now suppose that $W_n  \le^{icx} W'_n$ for a given $n$. We show below that this implies $W_{n+1} \le^{icx}  W'_{n+1}$.

From Theorem~\ref{thm:Egalitarianism} we have that $W_n + cs_n \le^{icx} W_n + cs'_n$. Further, by Lemma~\ref{lemma:MonotonicityOfBR} we have $W_n + cs'_n \le^{icx} W'_n + cs'_n$. Thus, we have $W_n +cs_n \le^{icx} W'_n +cs'_n$. Since $\mu \tau_n {\bf 1}$ has equal entries and $\max(.,0)$ is an increasing and convex operation, we have  $(W_n +cs_n - \mu\tau_n{\bf 1})^+ \le^{icx} (W'_n +cs'_n - \mu\tau_n{\bf 1})^+$, \ie, $W_{n+1} \le^{icx}  W'_{n+1}$. Hence the result holds. 
\end{proof}

The above theorem implies, for example, that each raw moment of the workload at given server under $\mathcal{WF}$ and $\mathcal{BS}$ is less than or equal to that under $\mathcal{BR}$. Similarly, each raw moment of the total workload in the system is lower or equal under $\mathcal{WF}$ and $\mathcal{BS}$ as compared to that under $\mathcal{BR}$. 

However, the above theorem does not directly allow us to compare the delays of requests for each $n$. To see this, recall that delay seen by a request is the $\max$ of the delays in downloading individual blocks, which are random in number. Further, a more unbalanced workload $W'_n$ may have more empty servers than $W_n$. The next arrival could, for example, have the associated blocks stored on the servers which are empty in $W'_n$ and not in $W_n$.
 
%
%

The following theorem compares delays of requests under both the policies. 

\begin{theorem}\label{thm:DelayComparision}
Consider a system which starts empty.  The delays seen by requests under the $\mathcal{WF}$, $\mathcal{BS}$, and $\mathcal{BR}$ policies satisfy the following:
 $$D^{ \Scale[0.6]{\mathcal{WF}} }_n \le^{icx}  D^{ \Scale[0.6]{\mathcal{BR}} }_n \text{ and } D^{ \Scale[0.6]{\mathcal{BS}} }_n \le^{icx}  D^{ \Scale[0.6]{\mathcal{BR}} }_n, \;\;\; n=0,1,\ldots$$
\end{theorem}
\begin{proof}
We show this for $\mathcal{BS}$; the argument for $\mathcal{WF}$ is analogous.

For ease of notation, we will use the notation $W_n, s_n$ for random vectors associated with policy $\mathcal{BS}$ with their usual meaning, and $W'_n, s'_n$ for those associated with policy $\mathcal{BR}$. 

For a given $\mathbb R_+^m$ valued vector $r$,  the function $\max_{i: r_i >0 } (x^i + r^i)$ is increasing and convex in $x \in \mathbb R^m$. Thus, for any increasing convex function $g: \mathbb R \to \mathbb R$, $g\left(\max_{i: r_i >0 } (x^i + r^i)\right)$ is an increasing convex function in $x$. Thus, from Theorem~\ref{thm:Workload_comparision} we have
$$
E_{s'_n}g\left(\max_{i: s^{\prime i}_n > 0} (W_n^i+c s^{\prime i}_n) \right) \le E_{s'_n}g\left(\max_{i: s^{\prime i}_n > 0}  ({W}_n^{\prime i}+c s^{\prime i}_n) \right),
$$
where $E_{s'_n}$ denotes the conditional expectaion given $s'_n$.
Note that on both sides of the above inequality we are conditioning on $s'_n$ which is the routing vector associated with $\mathcal{BR}$. 

Recall that under the $\mathcal{BR}$ policy, $s_n$ is independent of the instantaneous workload $W_n$ for each $n$. Using the coupling $\kappa_n = \kappa'_n$ and $a_n=a'_n$, and the definition, given instantaneous workload $W_n$ one can additionally couple the routing vectors $s_n$ and $s'_n$ and the associated $\kappa_n$ block requests under $\mathcal{BS}$ and $\mathcal{BR}$ policies such that the workload seen by the $l\kth$ block in front of it under $\mathcal{BS}$ is lower than that under $\mathcal{BR}$ for each $l\le \kappa_n$ under $W_n$. Thus, we get 
$$
E_{s'_n}g\left(\max_{i: s^i_n > 0}  (W_n^i+ cs^i_n) \right)\le
E_{s'_n}g\left(\max_{i: s^{\prime i}_n > 0}  (W_n^i+c {s'}^i_n) \right).
$$
By combining the previous two inequalities we get
$$E_{s'_n}g\left(\max_{i: s^i_n > 0} (W_n^i+cs^i_n) \right)
\le E_{s'_n}g\left(\max_{i: s^{\prime i}_n > 0}  ({W'}_n^i+c {s'}^i_n) \right)
,$$
from which the result follows by taking expectation on both sides. 
\end{proof}

Recall that $\rho = \lambda \nu/m = c \lambda \frac{\sum_k k\pi_k}m$ is the load factor per server.
The overall system load is $\rho m$.
By exchangeability, the marginal dynamics of the workload at a given server under $\mathcal{BR}$
can be modeled via an $M/GI/1$ FCFS queueing system with load $\rho$ bits/sec and service rate $\mu$ bits/sec. 
Since the number of servers $m$ is finite, the system is stable (asymptotically stationary) if $\rho < \mu$.
From Theorem~\ref{thm:Workload_comparision} and the ergodicity of the arrival process,
it follows that the system is stable under $\mathcal{WF}$ and $\mathcal{BS}$ as well if $\rho < \mu$. 

Note that,
for general $\alpha_k$,
the delays under the $\mathcal{BR}$ policy
are statistically equivalent to the delays
obtained when $\alpha_k = k$ for each $k$, \ie, when the code rate is
equal to $1$. There are prior works which study gains
of erasure-coding via simulations \cite{SLR14,JLS14},
experiments\cite{JDP05,LiK14}, and analytically but under
mean-filed type asymptotic approximations and under
exponential service time assumptions \cite{LRS16}.
To the best of our knowledge, Theorem \ref{thm:DelayComparision}
is the first rigorous analytical result which compares delays
for finite systems employing erasure codes with different code rates.
Further, we would like to stress that the result holds under
general statistical assumptions for service requirements.

\section{Association and Delay Bounds}\label{sec:Bounds}
In this section, we use the notion of association of random variables to obtain computable bounds on the delays of requests. 

\begin{definition}
The random variables $X_1,X_2,\ldots,X_k$ are {\em associated} if, with notation $X=(X_1,X_2,\ldots,X_k)$, the inequality 
$$  E[f(X)g(X)] \ge E[f(X)]E[g(X)]$$
holds for each pair of increasing functions $f,g:\mathbb R^k \to R$ for which $E[f(X)]$, $E[g(X)]$, and $E[f(X)g(X)]$ exist. 

We say that a random vector $X$ is associated if its entries are associated. Similarly, we say that a set of random variables is associated if its elements are associated. 
\end{definition}

To understand the power of association, consider the following definition and subsequent proposition. 

\begin{definition}
Consider random variables $X_1,X_2,\ldots,X_k$. We say that $\tilde X_1, \tilde X_2, \ldots, \tilde X_k$ are independent versions of the random variables $X_1,\ldots,X_k$ if the $\tilde X_1,\ldots, \tilde X_k$ are mutually independent, and if $X_i$ and $\tilde X_i$ are identically distributed for $1\le i \le k$. 
\end{definition}

\begin{proposition}[see \cite{BaB03} Chap 4.3] \label{prop:ComparingMax}
Suppose that random variables $X_1,\ldots,X_k$ are associated and that $\tilde X_1, \ldots, \tilde X_k$ are their independent versions. Then the following holds:
$$\max_{1\le i \le k} X_i \le^{st} \max_{1\le i \le k} \tilde X_i.$$
\end{proposition}

Now consider $m$ different queues with dependent workloads, as in the previous section. If we can show that the arrival of a request sees associated workloads, then we can bound its delay by using the independent version of the workloads. Several works in the literature for large-scale systems, \eg~\cite{BLP10,LRS16}, consider the marginal distribution at a given server and study its properties by assuming that the dynamic at any other server is independent of that under the given server; an assumption which is justified in these works as a `mean-field approximation'. In \cite{BLP10}, the queue associated with a given server is called a `queue at the cavity'. With the association property, can analyze a system without resorting to the mean-field approximation.

Recall that under the $\mathcal{BR}$ policy, the selection of servers $s_n$ is independent of the workload $W_n$. Upon an arrival, a server gets no additional workload with probability $1- \sum_{k=1}^{m} \frac{k}{m} \pi_k - \sum_{k = m+1}^{\infty} \pi_k$, and gets workload which is a multiple of $c$ otherwise.  One can show that, given that the request is of size $kc$, the server gets the load $c\left(\floor{\frac{k}{m}}+1\right)$ with probability $\frac{k}{m} - \floor{\frac{k}{m}}$ and the load $c\floor{\frac{k}{m}}$ with probability $1-\frac{k}{m} + \floor{\frac{k}{m}}$. 
Thus, for $i=1,\ldots,m$, the workload process at $i\kth$ server, namely $\{W^i_n\}_{n=0}^\infty$, in isolation is stochastically equivalent to workload seen by arrivals in a {\em Cavity Queue} which is as defined below. 

\begin{definition}\label{def:CavityQueue}
A {\em Cavity Queue} is an $M/GI/1$ FCFS queue which starts empty at time $t=0$, has Poisson arrivals with rate $\lambda m$, service rate $\mu$ bits/sec, and service requirement in bits with probability mass function on set $\{0,c,2c,\ldots\}$ given as follows:
$$\tilde \pi(0) = 1- \sum_{k=1}^{m} \frac{k}{m} \pi_k - \sum_{k = m+1}^{\infty} \pi_k,$$
and for $l=1,2,\ldots$
$$\tilde \pi(lc)   =    \sum_{k = (l-1)m+1}^{lm} (\frac{k}{m} - l+1) \pi_k + \sum_{k = lm+1}^{(l+1)m-1} (1-\frac{k}{m} + l) \pi_k.$$
\end{definition}

The $M/GI/1$ FCFS queues are well studied in the literature. In particular, the following lemma well-known as Pollaczek-Khinchine formula describes the steady state workload distribution of jobs in these queues. Below, we view service time of a job as the ratio of its service requirement in bits and the service rate of the server in bits/sec. 

\begin{lemma}[\!\!\cite{Tak62}]\label{lemma:PollaczekKhinchine}
Consider an $M/GI/1$ FCFS queue with arrival rate $\tilde\lambda$.
Let $\sigma$ be a random variable with distribution equal to
that of the service times of jobs. Let $\psi_\sigma(s) = E[e^{-s\sigma}]$. Suppose that $\tilde\lambda E[\sigma] < 1$. In steady state the workload $W$ has Laplace Transform $\mathcal{G}(.)$ (\ie, $\mathcal{G}(s) = E[e^{-sW}]$) which can be given as: 
\begin{equation}\label{eq:pkb}
 \mathcal{G}(s) = \frac{(1-\tilde\lambda E[\sigma]) s}{s-\tilde\lambda\left(1-\psi_\sigma(s)\right)}.
\end{equation}
\end{lemma}
Below, we use (\ref{eq:pkb}) to obtain
performance bounds on the systems of our interest by using
association property along with Proposition~\ref{prop:ComparingMax}. 
The following subset of the many known properties of association can come handy in proving association of random variables (RVs). 
\begin{proposition}[see \cite{BaB03} Chap 4.3]\label{prop:AssociationProperties}
The following statements hold.
\begin{enumerate}[(i)]
\item The set consisting of a single RV is associated. 
\item The union of independent sets of associated RVs forms a set of associated RVs
\item Any subset of a set of associated RVs forms a set of associated RVs
\item  For a non-decreasing function $\phi: \mathbb R^m \rightarrow \mathbb R$ and 
 associated RVs $\{X_1,\ldots,X_m\}$, the random variables 
$$\{\phi(X_1,\ldots,X_m),X_1,\ldots,X_m\}$$
are associated. 
\end{enumerate}
\end{proposition}

Before providing our main results for this section,
we need the following additional notation.

\begin{definition}\label{def:ConditionedDelay}
For each $k,n \in \mathbb Z_k$, let $D_n^k$ denote the delay seen by the $n\kth$ arrival given that the size of the requested file is $kc$ bits, that is, 
\begin{multline}
Pr\left( D_n^k \le t\right) = Pr\left(\mathop{\max}_{i: s_n^i > 0}  W_n^i+ c s_n^i  \le t \Big| \kappa_n = k \right), \\ t \in \mathbb R  \text{ and } k,n \in \mathbb Z_+.
\end{multline}
\end{definition}

Recall that for each $k,n\in \mathbb Z_+$, the $n\kth$ request for a file is of size $kc$ bits with probability $\pi_k$ and the $k$ requests for coded blocks are routed to different servers upon the $n\kth$ arrival as per the chosen policy.

\begin{definition}\label{def:Theta_m}
Let $\Theta(m)$ be the class of probability mass functions $\{\pi \}$ such that for each $\pi = (\pi_k:k\in\mathbb Z_+)$ in class $\Theta$, a system with $m$ servers operating under $\mathcal{BR}$ policy has the routing vector $s_n$ which is associated for each $n$. 
\end{definition}

 In this paper we will be content to note that $\Theta(m)$ is a rich class of p.m.f.s which includes $Binomial(p,m)$ distribution as well as $Geometric(p)$ distribution for each $p \in [0,1]$. 

The following theorem, proved in the Appendix,
says that for any $\pi$ in $\Theta(m)$, we get an upper bound on
the delay seen by the $n\kth$ arrival by pretending that the
workloads at the $m$ servers `evolved independently in the past'.

\begin{theorem}\label{thm:DelaysViaAssociation}
Consider a system with $m$ servers which starts empty. For each $k\in \mathbb Z_+$, requests for files of size $kc$ bits, equivalently batch requests for $k$ blocks of $c$ bits each, arrive as per an independent point process with rate $\pi_k \lambda m$ and are routed to different servers upon arrival. Servers serve the block requests in FCFS fashion at rate $\mu$ bits per second. 

Suppose that $\pi=(\pi_k:k\in\mathbb Z_+)$ belongs to class $\Theta(m)$ (see Definition~\ref{def:Theta_m}). Then the following statements hold:
\begin{enumerate}
\item 
The workload $W_n\in \mathbb{R}^m$, 
at the $m$ servers seen by the $n\kth$ arrival
under $\mathcal{BR}$ is associated for each $n$. 
\item For $i=1,2,\ldots,m$, let $\{\tilde W_n^i\}_0^\infty$ represent the workload seen by arrivals in an independent Cavity Queue as in Definition~\ref{def:CavityQueue}. Let $^k\!s$ be a typical routing vector $s$ under $\mathcal{BR}$ subject to $| s| = k$.  Under either $\mathcal{WF}$, $\mathcal{BS}$,
or $\mathcal{BR}$, the conditional delay $D_n^k$ of Definition \ref{def:ConditionedDelay} satisfies the following property: for each $k,n \in \mathbb Z_+$:
\begin{equation}
D_n^k \le^{icx} \mathop{\max}_{i:\; ^k\!s^i>0}  \tilde W_n^i+ c \; ^k\!s^i.
\end{equation}
\end{enumerate}

\end{theorem}

Here is now a uniform bound in $n$.

\begin{theorem}\label{thm:BoundViaSteadyState}
Consider a system satisfying
the assumptions of Theorem \ref{thm:DelaysViaAssociation}.
Suppose that $\rho = c \lambda \sum_{k=0}^{\infty}{k\pi_k}/m < \mu$. 
For $i=1,2,\ldots,m$, let $\tilde W^i$ represent the stationary
workload of an independent Cavity Queue.
Then,
under either $\mathcal{WF}$, $\mathcal{BS}$, or $\mathcal{BR}$,
the conditional delay $D_n^k$ satisfies the following property: 
for each $k,n \in \mathbb Z_+$:
\begin{equation}
D_n^k \le^{icx} \mathop{\max}_{i:\; ^k\!s^i>0}  \tilde W^i + c \; ^k\!s^i.
\end{equation}
\end{theorem}
\begin{proof}
This follows from Theorem~\ref{thm:DelaysViaAssociation} and noting that,
using standard coupling arguments, an $M/GI/1$ queue starting
empty at time $t=0$ and its version in equilibrium can be coupled
in such a way that the former is always lower than the latter.
\end{proof}

The above bound clearly reflects the impact of the local dynamics
at individual servers as well as the global view seen by arrivals.
As we shall see, it can be computed using Lemma \ref{lemma:PollaczekKhinchine}
and using extremal statistics. 

In what follows, we focus on $\pi$ such that $\pi_k =0$ for each $k>m$. 
Such a case is perhaps meaningful for clusters with very large $m$
since files which span each of the thousands of servers may be rare.
Under this scenario, Corollary \ref{cor:ViaMD2} below shows that delays 
admit a particularly simple bound. 

\begin{corollary}\label{cor:ViaMD2}
Consider a system with $m$ servers.
Suppose that $\pi$ belongs to class $\Theta(m)$ and that
$\pi_k = 0$ for each $k > m$. Suppose that $\rho < \mu$.
Let
\begin{equation}
q=q(\lambda,\sigma)=\left| \lambda 
+\frac{W\left(-\lambda\sigma \exp(-\lambda\sigma)\right)}{\sigma}\right|,
\label{eqlamb}
\end{equation}
where $W$ denotes the principal branch of the Lambert W function.
Then, under $\mathcal{WF}$, $\mathcal{BS}$
or $\mathcal{BR}$, the conditional steady state delay $D^k$ satisfies 
\begin{equation}
\label{eq:blam} 
E[D^k] -c \le \frac 1 {q(\frac \rho c, \frac c \mu)} \log k(1+o(1)),
\end{equation}
as $k$ tends to infinity, where $q$ is the function defined in (\ref{eqlamb}).
\end{corollary}
Note that the last relation implies that
$$ E[D^k]\le \frac 1 {q(\frac \rho c, \frac c \mu)} \log k(1+o(1)),$$
when $k$ tends to infinity. However it turns out that
the formulation in (\ref{eq:blam}) is numerically more accurate
in the prelimit.

Surprisingly, as long as $\pi$ belongs to $\Theta(m)$
and the load per server is fixed,
the above bound does not depend on $\pi$.
However, note that the bound is for 
the conditional delay. The bound on the overall delay still depends on $\pi$. 

This bound scales linearly with $c$ but logarithmically with $k$.  
Thus, for small and medium files, it pays to have smaller chunk size
(see Subsection \ref{sec:ICS} for a quantification of this gain).
This insight also concurs with the results obtained
in \cite{LRS16} under a mean field approximation.


\section{Random Chunk Sizes}\label{sec:RandomChunkSize}
We now study the scenario where the chunk size
may be different for different files, which is
modeled via recursion \eqn{eq:Main_Recursion2}.
Suppose that the random variables $\{c_n\}_0^\infty$
are i.i.d.\ with distribution $\psi$. 
The results of Section \ref{sec:Comparisons} 
readily extend to this scenario. In particular
the statement of Theorems~\ref{thm:Egalitarianism},
~\ref{thm:Workload_comparision} and \ref{thm:DelayComparision} 
can be shown to hold for this scenario as well,
with minor modifications in the proofs.
We skip the details for brevity. 

We now extend the results of Section \ref{sec:Bounds}.
We first modify the notion of Cavity Queue as follows. 

\begin{definition}\label{def:ModifiedCavityQueue}
The {\em Modified Cavity Queue} is an $M/GI/1$ FCFS
queue which starts empty at time $t=0$,
has Poisson arrivals with rate $\lambda m$,
service rate $\mu$ bits/sec, and where service requirement
in bits are \iid\ with distribution equal to that of the 
random variable $X$, where $X$ can is generated as follows:
first, generate a $\mathbb Z_+$ valued random variable $Y$
with probability mass function given as follows:
$$\tilde \pi(0) = 1- \sum_{k=1}^{m} \frac{k}{m} \pi_k - \sum_{k = m+1}^{\infty} \pi_k,$$
and for $l=1,2,\ldots$
$$\tilde \pi(lc)   =    \sum_{k = (l-1)m+1}^{lm} (\frac{k}{m} - l+1) \pi_k + \sum_{k = lm+1}^{(l+1)m-1} (1-\frac{k}{m} + l) \pi_k.$$
Let $Z$ be a random variable with distribution $\psi$. Then, $X = YZ$. 
\end{definition}

Recall that the steady state workload distribution
of an $M/GI/1$ FCFS queue satisfies Lemma~\ref{lemma:PollaczekKhinchine}. 
By using the above notion of Modified Cavity Queue,
analogues of Theorem~\ref{thm:DelaysViaAssociation}
and \ref{thm:BoundViaSteadyState} can be shown to hold
with minor modifications in proofs. Here, we only
reproduce the analogue of Theorem \ref{thm:BoundViaSteadyState} for brevity. 

\begin{theorem}\label{thm:BoundViaSteadyState2}
Consider a system with $m$ servers which start empty. The chunk sizes $\{c_n\}_0^\infty$ are i.i.d.\ with distribution $\psi$. For each $k\in \mathbb Z_+$, batch requests for $k$ blocks (\ie, coded chunks) arrive as per an independent point process with rate $\pi_k \lambda m$ and are routed to different servers upon arrival. Servers serve the block requests in FCFS fashion at rate $\mu$ bits per second. 

Suppose that $\pi=(\pi_k:k\in\mathbb Z_+)$ belongs to class $\Theta(m)$ (see Definition~\ref{def:Theta_m}).
Suppose that $\rho = E[c_1] \lambda \sum_{k=0}^{\infty}{k\pi_k}/m < \mu$. 
For $i=1,2,\ldots,m$, let $\tilde W^i$ represent the stationary
workload of an independent Modified Cavity Queue (see Definition~\ref{def:ModifiedCavityQueue}).
Then,
under either $\mathcal{WF}$, $\mathcal{BS}$, or $\mathcal{BR}$,
the conditional delay $D_n^k$ satisfies the following property: 
for each $k,n \in \mathbb Z_+$:
\begin{equation}
D_n^k \le^{icx} \mathop{\max}_{i:\; ^k\!s^i>0}  \tilde W^i + c \; ^k\!s^i,
\end{equation}
where $c$ is a random variable with distribution $\psi$. 
\end{theorem}

Again consider a scenario where $\pi_k =0$ for each $k>m$. Suppose that $\psi$ is exponential. Then the Modified Cavity Queue is an $M/M/1$ queue. Thus, the following corollary readily follows from the above theorem. 

\begin{corollary}\label{cor:ViaMD2}
Consider a system with $m$ servers. Suppose that $\pi$ belongs to class $\Theta(m)$, and that
$\pi_k = 0$ for each $k > m$. Suppose that the distribution $\psi$ is exponential with mean $c$. 
 Suppose that $\rho < \mu$.
Then, under $\mathcal{WF}$, $\mathcal{BS}$
or $\mathcal{BR}$, the conditional steady state delay $D^k$ satisfies 
\begin{equation}
\label{eq:blam} 
E[D^k] -c \le \frac{\mu}{\mu-\rho} \sum_{l=1}^{k} \frac{1}{l} < \frac{\mu}{\mu-\rho} (\log k + 1).
\end{equation}
\end{corollary}



\section{Simulation and Performance Evaluation}
\label{sec:sim-rec}

In this section we use our analysis and
simulations in order to develop a better
quantitative understanding of the relative performance
and scaling laws under $\mathcal{WF}$, $\mathcal{BS}$, and $\mathcal{BR}$.

\subsection{Simulation Methodology}

The simulation methodology we selected is not based on the classical discrete
event principles but rather on a direct use of the recurrence 
equations (\ref{eq:Main_Recursion}).
The advantages of the latter on the former
are multiple, in term of generality and of complexity.
This recurrence relation setting is well adapted to handling
deterministic service times and
general routing vectors, whereas 
event driven Markov chain simulation would require
exponentiality assumptions and make the handling of workload based
routing policies cumbersome.
The complexity of $\mathcal{BS}$ is that of a sorting algorithm.  If
the servers containing at most one chunk from the requested file are
sorted in increasing order of their load, then it suffices to take
the $k$ smallest loads if $k\leq m$. When $k>m$, the complexity 
depends on $m$ rather than $k$, as only $k-m\lfloor k/m\rfloor$
servers with the smallest load need to be searched. The complexity is
then in $O(\min(k,m)\log\min(k,m))$. 
The complexity of $\mathcal{WF}$ depends on $k$
whatever its value: one strategy is to first sort the servers
containing at least one chunk of the requested file.
Each time a chunk is requested from one server, its
load increases by $c$, and this server has to be re-inserted in the
ordered list of servers. The complexity 
is then $O(k\log(m+k))$.

In several experiments, the size of the files is at most $mc$ and
$\alpha_k-k\leq 2$. When $\alpha_k\leq m$ for all $k$, then $\mathcal{BS}$
and $\mathcal{WF}$ are exacly the same: as each server contains at most
one chunk of any given file, all the routing vectors are balanced,
and Theorem \ref{thm:Egalitarianism} states the optimality of 
$\mathcal{BS}$ in this case. For this reason, we
will only compare $\mathcal{BS}$ with $\mathcal{BR}$.


An important question is that of the steady state characterization.
For this, we leverage Birkhoff's pointwise ergodic theorem,
which shows that empirical averages based on iterates of the recurrence
equations (\ref{eq:Main_Recursion}) converge to the steady state mean values.
In practice, we perform $10^5$ iterates to estimate each point of the
following plots.

\subsection{Impact of the Delivery Policy}
The first numerical experiments illustrate the
comparison results of Section \ref{sec:Comparisons} and
more precisely Theorem \ref{thm:Workload_comparision}. 
The setting is the following: there are $m=200$ servers;
the distribution $\pi$ is {\it{Binomial}}$(m,p)$, with
$p=0.1, 0.3,0.5$ (which gives an average of 20, 60 and 100 chunks,
respectively); recall that for each value of $p$ this falls within the class of distribution $\Theta(m)$; the server speed is $\mu=1$ and the chunk size is
$ c=10$; the arrival rate is chosen in such a way that
the load per server is always equal to 0.7;
the coding assumptions are that $\alpha_k=k+2$.

Figure \ref{fig:filesize}
compares the mean delay under $\mathcal{BR}$ and $\mathcal{BS}$,
for various values of $p$. 
 The bound obtained in Corollary \ref{cor:ViaMD2}  is also plotted.
Within the range considered in these plots, 
the mean delays increase logarithmically in $k$ for $\mathcal{BR}$.
The bound correctly captures the logarithmic increase w.r.t.\ the 
$\mathcal{BR}$, and is in fact approximate for small $p$. For these parameters, it is already
a good heuristic for $p=0.1$.

\begin{figure}[ht]
\centering
\includegraphics[width=0.5\textwidth]{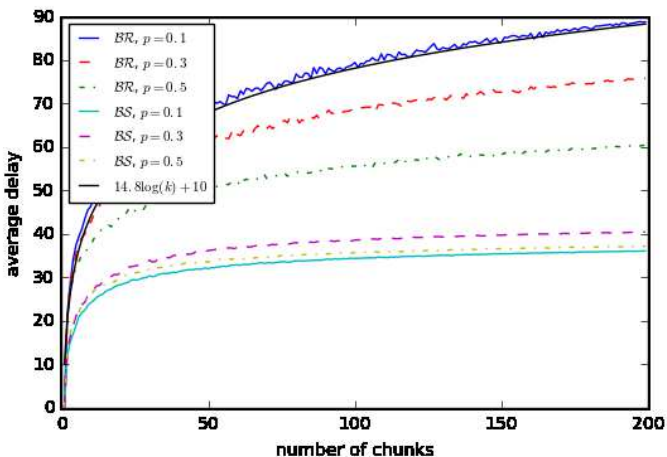}
\caption{Mean delay as a function of the number of chunks.}
\label{fig:filesize}
\end{figure}

We observe that $\mathcal{BS}$ (or equivalently $\mathcal{WF}$)
performs significantly better than $\mathcal{BR}$. Intuitively, this happens since 
the workload across servers is more balanced under $\mathcal{BS}$ and $\mathcal{WF}$. 
 In particular,  while 
they seem to increase as $\log \log k$ for $\mathcal{BR}$ and $\mathcal{WF}$. 
One may see this in the light of the well-known result on balanced allocations under balls and 
bins setting \cite{ABK00} where load-balancing is shown to achieve exponential improvement in
load at the most-loaded bin. However, our setting is markedly different. Not only do we 
incorporate queuing dynamics (\ie, arrivals and services), but also batch arrivals. 
We are interested in studying the delay of a typical job which depends on the workload at a 
randomly chosen subset of servers, instead of the most-loaded server.

\begin{figure}[ht]
\centering
\includegraphics[width=0.5\textwidth]{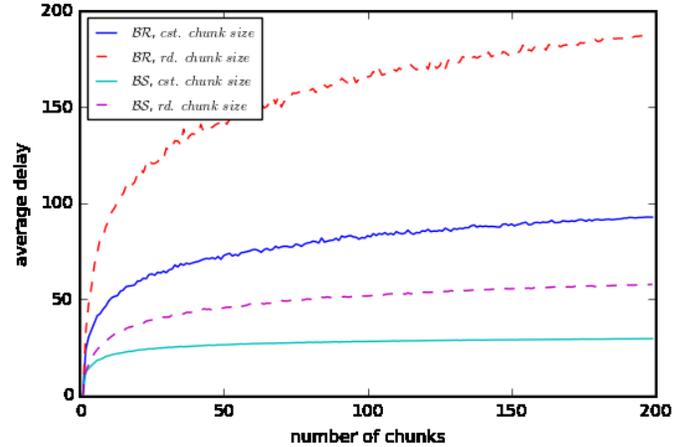}
\caption{Mean delay as a function of the number of chunks: Comparison of scenarios with constant chunk size and random chunk size.}
\label{fig:filesize_randomchunk}
\end{figure}

Interestingly, in our setting, the increase in delays seems logarithmic in $k$ 
even for $\mathcal{BR}$ and $\mathcal{WF}$ policies
under a scenario where the chunk size is assumed to be random, as 
exhibited in Figure~\ref{fig:filesize_randomchunk}. The setting is the following: there are $m=200$ servers;
the distribution $\pi$ is Geometric with rate $0.25$, the size of chunks are
exponentially distributed with rate $0.1$. 
 The load per server is $0.7$. 
The coding assumptions are that $\alpha_k=k+2$. The plots show that, 
in the cases where the chunks sizes are exponentially distributed, the $\log$
growth in delays as exhibited by the upper bound of Corollary~\ref{cor:ViaMD2} is tight when the per-server load is sufficiently large.

Under the assumptions studied above, for each policy,
the growth of delays is logarithmic or sub-logarithmic in the file size.
This type of growth does not generalize to all cases.
For instance, it is shown in Subsection \ref{sec:BLR} below
that it can actually be linear.


\subsection{Impact of Coding Rate}
In order to evaluate the impact of coding rate, we consider a system
under $\mathcal{BS}$ with $m$ servers, where $m$ varies.
We take $\lambda = 0.1$, $p=0.5$ and again $\pi$ is
{\it{Binomial}}$(p,m)$, so that the load per server is constant.
We take $c=14$ and $\mu=1$.
Figure \ref{fig:delay2} gives the mean delay as a function of $m$
for different choices of $\alpha_k - k$.

As expected, $\mathcal{WF}$ and $\mathcal{BS}$ perform significantly
better than $\mathcal{BR}$ when $\alpha_k > k$. We observe that
the delays increase logarithmically with $m$. This may be reasoned as 
follows: In the presence of small and medium sized files if $\alpha_k - k$ 
is a constant then the choice in load-balancing is 
limited and the unevenness in workload distribution across servers 
increases with $m$. Further, 
as we increase the code redundancy $\alpha_k - k$, 
we observe that the mean delays decrease as $\frac{1}{\log(\alpha_k - k)}$. 
This shows that the impact of increasing choice in load-balancing by 
improving coding rate is limited.

\begin{figure}[ht]
\centering
\includegraphics[width=0.5\textwidth]{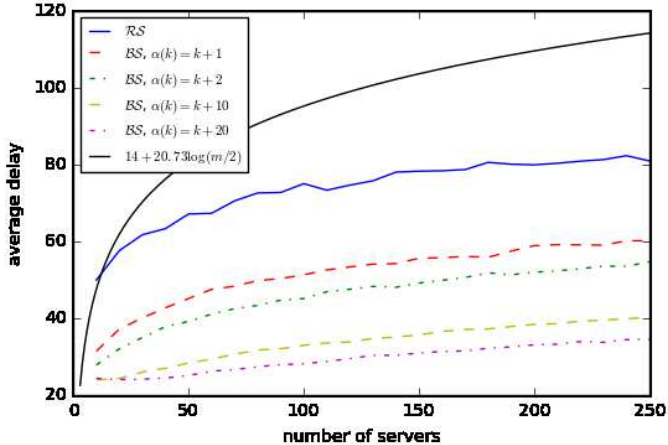}
\caption{Mean delay under $\mathcal{BS}$ as a function of 
number of servers, for different coding rates.}
\label{fig:delay2}
\end{figure}

\subsection{Impact of the (Deterministic) Chunk Size}
\label{sec:ICS}
We now consider the impact of increasing chunk size on delays for  the
case with $\pi$ being {\it{Binomial}}$(p,m)$.
Rather than taking chunks of size $c$, we take chunks
of size $c/a$ with $a$ an integer larger than 1, and study 
mean delay as a function of $a$. 
Here, a file which had $k$ chunks now has $ak$ chunks.
Consider the upper bound of 
Corollary \ref{cor:ViaMD2} (this bound is generic in that 
it holds for all considered delivery policies).
The bound in the new chunk definition is now 
$ \frac 1 {|s^*(a)|} \ln(1+k) (1+o(1))$,
when $k$ tends to infinity, with $s^*(a)$ the only negative
solution of the equation
$ s= \lambda p a\left(1-\exp\left(-\frac {sc}{a}\right)\right).$
When $a$ is large (but such that $pa<1$), this root can be approximated as
$ |s^*(a)|= \frac {(1-\lambda pc)2a} {c^2\lambda p}.$
so that we have the generic bound on requests of initial cardinality $k$:
$$ E[D^k]\le \frac{c^2\lambda p}{(1-\lambda pc)2a} \ln(ak)(1+o(k)),$$
when $k$ tends to infinity.
This shows that within the above
Binomial setting, the mean delay of any policy
can be decreased in such a way that the constant multiplying
the logarithmic term is divided by $a$ (provided $p a<1$). 
%

%
%

\subsection{Beyond the Logarithmic Regime}
\label{sec:BLR}
The last three subsections were about the case where 
$\pi$ has its support on the integers from 0 to $m$.
In view of the results of these subsections, it makes sense
to call this regime the {\em logarithmic regime}. There are some
caveats with this terminology. This term is justified
within the ${\it{Binomial}}(p,m)$ setting, if $p$ is sufficiently
separated from 1. As we saw above,
for $p$ constant and less than 1, the logarithmic regime prevails
even when $m$ tends to infinity. Note that this goes way beyond the regimes
considered in the mean field approach.
However, it should be clear that for fixed $m$ and 
for $p$ close to 1,
the mean delay must be approximately a constant in $k$.

Note that when the support of $\pi$ is 
not limited to the integers less than $m$ with $m$ fixed,
it should be clear that for all delivery policies, when $k$ tends
to infinity, requests of cardinality $k$ have a mean delay of
order $Ck$ with $C$ a constant. This is the linear regime alluded to above.

\section{Conclusions}\label{sec:Conclusions}
One of the main motivations of this work was to derive
scaling laws for job delays in data clusters.
A primary difficulty in the analysis of job delays
in multi-server systems comes from the stochastic coupling
of the server dynamics.
To simplify the analysis, research often resorts to 
an asymptotic `mean-field' approximation which assumes 
an infinite number of servers and a 
static empirical distribution. 
This approximation allows for the decoupling of the dynamics at 
the servers attending a tagged job.
However, such a decoupling does not hold when the total number 
of servers $m$ is finite, or when certain jobs are attended by $O(m)$  
servers. 
In the present paper, we developed a new machinery which utilizes
the notion of association of random variables to obtain
explicit bounds on delays for {\em finite systems}. We 
obtain these bounds via an `independent version' of a coupled system but
without requiring the decoupling of the servers. Further, we clarified the
sense (increasing convex ordering) in
which adaptive policies outperform workload oblivious policies. 
Our simulation results suggest that several quite different delay growths
can be obtained in function of file size,
from strictly sub-logarithmic to logarithmic to linear.
While some specific examples of these
behaviors are well explained by our machinery,
there is still a need in the future for a full classification
allowing one to predict which assumptions lead to each type of growth.

Our machinery is robust to statistical assumptions and to
model specifics. In addition,
various types of file updates/writes operations
can be incorporated in the basic model while preserving the
basic association and stochastic comparison properties. 
In the future, this model should hence also provide a first
comprehensive setting for analyzing the impact of updates 
on job delays in data clusters.

%
%

\section*{Acknowledgements}
We would like to thank Alex Dimakis at The University of Texas at Austin, and Ankit Singh Rawat at Massachusetts Institute of Technology for helpful discussions.


\begin{thebibliography}{10}

\bibitem{VPK15}
A.~Verma, L.~Pedrosa, M.~R. Korupolu, D.~Oppenheimer, E.~Tune, and J.~Wilkes,
  ``Large-scale cluster management at {Google} with {Borg},'' in {\em
  Proceedings of the European Conference on Computer Systems (EuroSys)},
  pp.~18:1--18:17, 2015.

\bibitem{SKR10}
K.~Shvachko, H.~Kuang, S.~Radia, and R.~Chansler, ``The hadoop distributed file
  system,'' in {\em 2010 IEEE 26th Symposium on Mass Storage Systems and
  Technologies (MSST)}, pp.~1--10, 2010.

\bibitem{GHJ09}
A.~Greenberg, J.~R. Hamilton, N.~Jain, S.~Kandula, C.~Kim, P.~Lahiri, D.~A.
  Maltz, P.~Patel, and S.~Sengupta, ``{VL2}: A scalable and flexible data
  center network,'' in {\em Proceedings of the ACM SIGCOMM 2009 Conference on
  Data Communication}, pp.~51--62, 2009.

\bibitem{GGL03}
S.~Ghemawat, H.~Gobioff, and S.-T. Leung, ``The google file system,'' in {\em
  Proceedings of the Nineteenth ACM Symposium on Operating Systems Principles
  (SOSP)}, pp.~29--43, 2003.

\bibitem{FLP10}
D.~Ford, F.~Labelle, F.~Popovici, M.~Stokely, V.-A. Truong, L.~Barroso,
  C.~Grimes, and S.~Quinlan, ``Availability in globally distributed storage
  systems,'' in {\em Proceedings of the 9th USENIX Symposium on Operating
  Systems Design and Implementation}, pp.~61--74, 2010.

\bibitem{DGW10}
A.~Dimakis, P.~Godfrey, Y.~Wu, M.~Wainwright, and K.~Ramchandran, ``Network
  coding for distributed storage systems,'' {\em Information Theory, IEEE
  Transactions on}, vol.~56, pp.~4539--4551, Sept 2010.

\bibitem{PlH13}
J.~S. Plank and C.~Huang, ``Tutorial: Erasure coding for storage
  applications.'' Slides presented at FAST-2013: 11th Usenix Conference on File
  and Storage Technologies:
  \url{http://web.eecs.utk.edu/~plank/plank/papers/FAST-2013-Tutorial.html},
  February 2013.

\bibitem{ShV15_TON}
V.~Shah and G.~de~Veciana, ``High performance centralized content delivery
  infrastructure: Models and asymptotics,'' {\em IEEE/ACM Transactions on
  Networking}, vol.~23, pp.~1674--1687, Oct 2015.

\bibitem{ShV16_Ind}
V.~Shah and G.~de~Veciana, ``Asymptotic independence of servers' activity in
  queueing systems with limited resource pooling,'' {\em Queueing Systems},
  vol.~83, no.~1, pp.~13--28, 2016.

\bibitem{ShV16_Het}
V.~Shah and G.~de~Veciana, ``Impact of fairness and heterogeneity on delays in
  large-scale centralized content delivery systems,'' {\em Queueing Systems},
  vol.~83, no.~3, pp.~361--397, 2016.

\bibitem{Vve98}
N.~D. Vvedenskaya, ``Large queueing system where messages are transmitted via
  several routes,'' {\em Problemy Peredachi Informatsii}, vol.~34, no.~2,
  pp.~98--108, 1998.

\bibitem{LRS16}
B.~Li, A.~Ramamoorthy, and R.~Srikant, ``Mean-field-analysis of coding versus
  replication in cloud storage systems,'' in {\em Proceedings of IEEE INFOCOM},
  2016.

\bibitem{SLR14}
N.~B. Shah, K.~Lee, and K.~Ramchandran, ``The {MDS} queue: Analysing the
  latency performance of erasure codes,'' in {\em IEEE International Symposium
  on Information Theory (ISIT)}, pp.~861--865, June 2014.

\bibitem{JLS14}
G.~Joshi, Y.~Liu, and E.~Soljanin, ``On the delay-storage trade-off in content
  download from coded distributed storage systems,'' {\em IEEE Journal on
  Selected Areas in Communications}, vol.~32, pp.~989--997, May 2014.

\bibitem{JDP05}
S.~Jain, M.~Demmer, R.~Patra, and K.~Fall, ``Using redundancy to cope with
  failures in a delay tolerant network,'' in {\em ACM SIGCOMM}, (New York, NY,
  USA), pp.~109--120, 2005.

\bibitem{LiK14}
G.~Liang and U.~C. Kozat, ``{TOFEC:} achieving optimal throughput-delay
  trade-off of cloud storage using erasure codes,'' in {\em IEEE INFOCOM 2014 -
  IEEE Conference on Computer Communications}, pp.~826--834, April 2014.

\bibitem{BaB03}
F.~Baccelli and P.~Br\'{e}maud, {\em {Elements of Queueing Theory: Palm
  Martingale Calculus and Stochastic Recurrences}}, vol.~26 of {\em
  Applications of Mathematics}.
\newblock Berlin: Springer-Verlag, second~ed., 2003.

\bibitem{BMS89}
F.~Baccelli, A.~M. Makowski, and A.~Shwartz, ``The fork-join queue and related
  systems with synchronization constraints: Stochastic ordering and computable
  bounds,'' {\em Advances in Applied Probability}, vol.~21, no.~3,
  pp.~629--660, 1989.

\bibitem{MuS02}
A.~M{\"u}ller and D.~Stoyan, {\em Comparison methods for stochastic models and
  risks}.
\newblock Wiley, 2002.

\bibitem{MOA11}
A.~W. Marshall, I.~Olkin, and B.~C. Arnold, {\em Inequalities: Theory of
  Majorization and Its Applications}.
\newblock Springer, 2nd~ed., 2011.

\bibitem{BLP10}
M.~Bramson, Y.~Lu, and B.~Prabhakar, ``Randomized load balancing with general
  service time distributions,'' in {\em Proceedings of the ACM Sigmetrics},
  pp.~275--286, 2010.

\bibitem{Tak62}
L.~Tak{\'a}cs, {\em Introduction to the theory of queues}.
\newblock University texts in the mathematical sciences, Oxford University
  Press, 1962.

\bibitem{ABK00}
Y.~Azar, A.~Z. Broder, A.~R. Karlin, and E.~Upfal, ``Balanced allocations,''
  {\em SIAM J. Comput.}, vol.~29, pp.~180--200, Feb. 2000.

\bibitem{Mit96}
M.~D. Mitzenmacher, {\em The Power of Two Choices in Randomized Load
  Balancing}.
\newblock PhD thesis, University of California, Berkeley, 1996.

\bibitem{VDK96}
N.~D. Vvedenskaya, R.~L. Dobrushin, and F.~I. Karpelevich, ``Queueing system
  with selection of the shortest of two queues: An asymptotic approach,'' {\em
  Problemy Peredachi Informatsii}, vol.~32, no.~1, pp.~20--34, 1996.

\bibitem{YSK15}
L.~Ying, R.~Srikant, and X.~Kang, ``The power of slightly more than one sample
  in randomized load balancing,'' in {\em Proc. of IEEE INFOCOM}, 2015.

\end{thebibliography}
\bibliographystyle{ieeetr}



\section{Appendix}

\subsection{Proof of Lemma~\ref{lemma:MajImpliesIcx}}
Recall that if $X\prec_s^{st} Y$ then $E[\phi(X)] \le E[\phi(Y)]$ for any increasing Schur-convex function $\phi$. Now consider an increasing convex function $g:\mathbb R^m \rightarrow \mathbb R$. Let $P$ be the set of all permutations of $(1,2,\ldots,m)$. One can check that for any $p\in P$, the function $g(p(x))$ is increasing and convex in $x$. Let function $\phi$ be given as follows: 
$$\phi(x) = \frac{1}{m!} \sum_{p \in P} g(p(x)).$$
Then, $\phi$ is a symmetric, increasing, and convex function; hence an increasing Schur-convex function \cite{MOA11}. 
Further, by exchangeability of $X$, we have $E[g(X)] = E[g(p(X))]$ for any $p \in P$, which in turn implies $E[g(X)] = E[\phi(X)] $. Similarly, by exchangeability of $Y$, we have $E[g(Y)] = E[\phi(Y)]$. But as noted above, we have $E[\phi(X)] \le E[\phi(Y)]$. 
%
%
The result thus follows since $g$ is chosen arbitrarily.

\subsection{Proof of Theorem~\ref{thm:Egalitarianism}}

We show $\prec^{st}$ comparisons below. The $icx$ comparisons then follow from Lemma~\ref{lemma:MajImpliesIcx} and noting that $W+s^\mathcal{WS}$, $W+s^\mathcal{BS}$, and $W+s^\mathcal{BR}$ are exchangeable since each of these policies is exchangeable. 

We will need below the following lemma, which says that a vector becomes more balanced if we decrease a larger entry by a `small' amount and increase a smaller entry by the same amount.

\begin{lemma}
  \label{lemma:transform}
  Let $x\in\Reals^m$ such that $x^i\leq x^j$ and $0\le\delta \leq x^j-x^i$. Then $x+\delta e_i-\delta e_j \prec  x$. 
\end{lemma}
\begin{proof}
  Set $y=x+\delta e_i-\delta e_j$. There exist $k$ and $l$ such that
  $x_i = x^{(k)}$ and $x_j = x^{(l)}$, with $k<l$, $k'$ and $l'$ such
  that $y_i = y^{(k')}$ and $y_j = y^{(l')}$, and as $\delta\leq
  x^j-x^i$, $k\leq k',l'\leq l$. 

  For all $i'<k$ and $i'>l$, we have $\sum_{u\leq i'} x^{(u)} =\sum_{u\leq i'} y^{(u)}
  $: in the first case, exactly the same terms are involved, and in
  the second, $x^i+x^j = y^i+y^j$, and these terms are all involved.

  If $k\leq i'< \min(k',l')$, then $\sum_{u\leq i'}y^{(u)} = \sum_{u < k}x^{(u)} +\sum_{k\leq u \leq i'}x^{(u+1)} \geq  \sum_{u < k}x^{(u)} +\sum_{k\leq u \leq i'}x^{(u)}$. 

  If $\min(k',l')\leq i' < \max(k',l')$, $\sum_{u\leq i'}y^{(u)} =
  \sum_{u < k}x^{(u)}+\sum_{k\leq u <\min(k',l')}x^{(u+1)} +
  \min(y^i,y^j) + \sum_{\min(k',l')< u \leq i' }x^{(u)} \leq \sum_{u
    \leq i'}x^{(u)}$, as $\min(y^i,y^j) = \min( x^i+\delta,x^j-\delta)
  \geq x^i$ (because $\delta \le x^j-x^i$).
  
  If $\max(k',l') \leq i' \leq l$, $\sum_{u\leq i'}y^{(u)} =
  \sum_{u\leq i'}x^{(u)} + x^{(l)} - x^{(i')} \geq \sum_{u\leq
    i'}x^{(u)}$ (we have used that $x^{(l)} = x^j \geq x^{(i')}$).
  
  Then $y\prec x$. 
\end{proof}

{\em Optimality of $\mathcal{WS}$:}
We now show that $\mathcal{WS}$ achieves more balanced workload than any other policy.
For ease of notation, let $s$, $\kappa$ and $a$ represent the vectors associated with $\mathcal{WS}$ with their usual meaning, and let $s'$, $\kappa'$ and $a'$ represent those associated with any other policy. Recall that the number of chunks for the requested file $\kappa$ and the placement vector $a$ associated with an arrival have same distribution in each system and they are independent of the workload seen by the arrival. Thus, it is sufficient to prove that $W+ cs \prec W + cs'$ w.p.\ $1$ subject to the coupling $\kappa = \kappa'$ and $a=a'$.

  We proceed as follows.  For any routing vector $s''$, define
  its distance to $s$ as $d(s,s'') = \sum_{i~|~s^i>s''^i}s^i-s''^i$. As
  $s$ and $s''$ are integer-valued, $d(s,s'')$ is a non-negative
  integer.  
  
Under the coupling $\kappa = \kappa'$ and $a = a'$,
  we show that for the routing vector $s'\neq s$, there exists
  another routing vector $s''$ such that $d(s,s'')<d(s,s')$ and
  $W+cs''\prec W+cs'$. This means that for any routing vector $s'$, we
  can construct a sequence $(s_0,\ldots,s_d)$ such that $W+cs_0 \prec
  W+cs_{1} \prec \cdots \prec W+cs_d = W+cs'$, with $d(s_0,s) = 0$,
  that is, $s=s_0$. In conclusion, for all routing vector $s'$, $W+cs
  \prec W+cs'$, hence the optimality of water-filling.

  Let us now prove the existence of the routing vector $s''$.
  Note that any routing vector less than
  $\max(s,s')$ is admissible, \ie, $\max(s,s') \le a$. As
  $s' \neq s$, there exists $i$ and $j$ such that $(W+cs)^i <
  (W+cs')^i$ and $(W+cs)^j > (W+cs')^j$, and as $s$ and $s'$ are
  integer-valued, $s^i\leq s'^i-1$ and $s^j\geq s'^j+1$.

  Consider the step of the water-filling algorithm where a chunk is
  sent to server $j$ for the last time, and let $\tilde{s}$ be the routing
  vector obtained just before the chunk is sent to server $j$ step: in
  particular $\tilde{s}^j = s^j-1$, so 
 \begin{equation}\label{eq:intermediate}
 (W+cs')^j\leq(W+cs)^j -c = (W+c\tilde{s})^j.
  \end{equation} 
  
  Due to the water-filling algorithm, server $j$ 
  is chosen over $i$ because $(W+c\tilde{s})^j \leq (W+c\tilde{s})^i$. But we also have $ (W+c\tilde{s})^i \leq (W+cs)^i\leq (W+cs')^i-c$. Thus, we get $ (W+c\tilde{s})^j \leq (W+cs')^i-c$. Combining this with \eqn{eq:intermediate}, we get $(W+cs')^j \le (W+cs')^i-c$. 

  Now, consider the new vector where $s'' = s' + e_j-e_i$. We
  have $d(s,s'') < d(s,s')$ and from Lemma~\ref{lemma:transform},
we have $W+ s''\prec W + s'$, as required.

{\em Comparing $\mathcal{BS}$ with $\mathcal{BR}$:} We now show that $\mathcal{BS}$ achieves more balanced workload than $\mathcal{BR}$. For ease of notation, let $s$, $\kappa$ and $a$ represent the vectors associated with $\mathcal{BS}$ with their usual meaning, and let $s'$, $\kappa'$ and $a'$ represent those associated with $\mathcal{BR}$. We again assume the coupling $\kappa = \kappa'$ and $a = a'$.

Under the coupling, we will show a statement which is somewhat stronger than required; in particular, we will show that the batch-sampling is optimal
  among all {\em balanced routing vectors}, \ie\ the routing
  vectors $s'$ such that $s'^i\in\{l,l+1\}$ where $l = \floor{\frac{\kappa'}{m}}$, and $|\{i:s'^i = l+1\}| = \kappa' - m\floor{\frac{\kappa'}{m}}$. 
  Moreover, due to our coupling $\kappa = \kappa'$ and $a = a'$, we only need to
  focus on routing vectors $s,s'$ of the form $\{0,1\}^m$.

We now proceed as follows: Take any balanced routing vector $s''$ and
define
  its distance to $s$, the routing vector obtained with the
  water-filling policy as $d(s,s'') = \sum_{i~|~s^i>s''^i}s^i-s''^i$. As
  $s$ and $s'$ are integer-valued, $d(s,s')$ is a non-negative
  integer.  We show that for any routing vector $s'\neq s$, there
  exists another balanced routing vector $s''\in\{0,1\}^m$ such that
  $d(s,s'')<d(s,s')$ and $W+cs''\prec W+cs'$. This means that for any
  routing vector $s'$, we can construct a sequence $(s_0,\ldots,s_d)$
  such that $W+cs_0 \prec W+cs_{1} \prec \cdots \prec W+cs_d = W+cs'$,
  with $d(s_0,s) = 0$, that is, $s=s_0$. In conclusion, for all
  routing vector $s'$, $W+cs \prec W+cs'$, hence the optimality of
  batch-sampling among the balanced routing vectors.

  Let us now prove the existence of the routing vector $s''$.  Note
  that any routing vector less than $\max(s,s')$ is admissible (there
  are enough chunks available). As $s' \neq s$, there exists $i$ and
  $j$ such that $s^i =0$, $ s'^i=1$, $s^j=1$ and $s'^j=0$, and as $s$
  is obtained from the batch-sampling, one can always such an $i$ and
  $j$ such that that $W^j \leq W^i$, so $(W+cs')^j \leq (W+cs')^i$ and
  $(W+cs')^j = W^j \leq W^i = X^i+cs'^i - c = (W+cs')^i - c$.

  Consider the routing vector $s'' = s' + e_j-e_i$.  From
  Lemma~\ref{lemma:transform}, the above inequality implies that
  $W+ s''\prec W+ s'$, and $d(s,s'') < d(s,s')$, as required.

\subsection{Proof of Lemma~\ref{lemma:MonotonicityOfBR}}

Since $W \le^{icx} W'$, Strassen's theorem \cite{MuS02} says that there exists a coupling such that $E[W'|W] \ge W$. In addition, since $s$ and $s'$ are identical in distribution and independent of $W$ and $W'$, there exists a coupling (namely one with $s = s'$) such that
\begin{equation}\label{eq:Mod_Strassen}
E\left[W' + cs' | W,s\right] \ge W + cs.
\end{equation}

%

Consider an increasing convex function $g$. Under the above coupling, using Jensen's inequality we get 
$$E[g(W'+cs')|W,s] \ge g\left(E[W'+cs'|W,s]\right).$$
Combining this with \eqn{eq:Mod_Strassen}, we get 
$$E[g(W'+cs')|W,s] \ge g(W+cs).$$
By taking expectation on both sides, we get $E[g(W'+cs')] \ge E[g(W+cs)]$. Hence the result holds. 

\subsection{Proof of Theorem \ref{thm:DelaysViaAssociation}}
We first prove part $(i)$ using induction. Clearly, $W_0$ is associated since all its entries are constant and equal to zero. Suppose $W_n$ is associated for some $n$. We show below that this implies that $W_{n+1}$ is associated as well. 

Recall that under $\mathcal{BR}$ the random vectors $W_n$, $s_n$, and $-\mu\tau_n{\bf 1}$ are mutually independent and are themselves associated. Hence, from part $(ii)$ of Proposition~\ref{prop:AssociationProperties}, we have that the entries of $W_n$, $s_n$, and $-\mu\tau_n{\bf 1}$ are mutually associated. Each entry of $(W_n + c s_n - \mu \tau_{n}  {\bf 1})^+$ is an increasing function of the entries of $W_n$, $s_n$, and $-\mu\tau_n{\bf 1}$. From $m$ applications of part $(i)$ of Proposition~\ref{prop:AssociationProperties},  and then of its part $(iii)$,  we get that $W_{n+1}$ is associated. 

We now prove part $(ii)$ of the theorem. We show that in fact for $\mathcal{BR}$ policy the stochastic dominance is in $\le^{st}$ sense which is stranger than $\le^{icx}$. By definition, the vector $\tilde W_n = (\tilde W_n^i: i=1,\ldots,m)$ is an independent version of $W_n$ for each $n$ under $\mathcal{BR}$ policy. Since $\tilde{s}$ is an independent exchangeable vector, and since both $W_n$ and $\tilde W_n$ are exchangeable, it is sufficient to assume that $\tilde{s}$ is deterministic. Then, the result follows for $\mathcal{BR}$ from Proposition \ref{prop:ComparingMax}. 

For $\mathcal{WF}$ and $\mathcal{BS}$, the result then follows by arguing along the lines of Theorem~\ref{thm:DelayComparision}  while additionally conditioning on $|s_n| = k$.

\subsection{Proof of Corollary~\ref{cor:ViaMD2}}
The proof leverages the following two results: 
\begin{theorem} [Theorem 7.4. in \cite{BMS89}]
Let $\{ Y_l \}_1 ^\infty $ be a family
of i.i.d. $\mathbb{R} _+ $-valued random variables whose common
distribution function $G(\cdot)$ exhibits the tail behavior
$$ P[Y_1 >x]= 1- G(x) = C e^{-qx}( 1 + o(1) ), \quad x \geq 0,$$
for some $q >0$ and $C>0$. Then
$$E \biggl [ \max \bigl \{ Y_1,...,Y_k \bigr \} \biggr ] =
\frac 1 q \log(k)(1+o(1))
$$
when $k$ goes to infinity.
\end{theorem}

\begin{lemma}
The steady state $Y$ delay in the M/D/1 queue with arrival rate $\lambda$
and service time $\sigma$, with $\lambda\sigma<1$ has the tail behavior 
$$P[Y >x]= 1- G(x) = C e^{-qx}( 1 + o(1) ), \quad x \geq 0,$$
with $q$ defined as in Equation (\ref{eqlamb}).
\end{lemma}

\begin{proof}
The Pollaczek-Khinchine formula, of Lemma \ref{lemma:PollaczekKhinchine},
when applied to the M/D/1 queue, gives 
a steady state delay with a Laplace transform having an
isolated pole at the only solution other than 0 of the equation
$$s=\lambda(1-\exp(-s\sigma)).$$
Elementary calculations show that this solution is precisely
$q$ given in (\ref{eqlamb}).
The shape of the tail then follows from classical complex
analysis arguments.
\end{proof}

The fact that the delay of a request of size $k$ 
is upper bounded by $c$ plus the maximum of the
workloads in $k$ independent M/D/1 queues
with arrival rate $\lambda p$ and service times $c$
immediately leads to the announced result.

\end{document}